\documentclass[a4paper,UKenglish,thm-restate]{article}
\usepackage[utf8]{inputenc}
\usepackage{amsthm}
\usepackage{amsmath}
\usepackage{amssymb}
\usepackage{mathtools}
\usepackage{hyperref}
\usepackage[nameinlink, capitalise]{cleveref}
\crefname{algorithm}{Algorithm}{Algorithms}
\usepackage{acronym}
\usepackage{csquotes}
\usepackage{bm}
\usepackage[boxruled,linesnumbered]{algorithm2e}
\usepackage{crossreftools}

\newtheorem{theorem}{Theorem}
\newtheorem{lemma}[theorem]{Lemma}
\newtheorem{observation}[theorem]{Observation}
\newtheorem{corollary}[theorem]{Corollary}
\newtheorem{proposition}[theorem]{Proposition}

\newtheorem{definition}[theorem]{Definition}

\newcommand{\ocal}{\ensuremath{\mathcal{O}}}
\newcommand{\dcup}{\ensuremath{\mathbin{\dot\cup}}}
\newcommand{\OPT}{\operatorname{OPT}}
\newcommand{\eps}{\varepsilon}

\DeclarePairedDelimiter\abs{\lvert}{\rvert}

\acrodef{TSP}[\textnormal{\textsc{tsp}}]{\textsc{Traveling Salesperson Problem}}
\acrodef{MDMTSP}[\textnormal{\textsc{mdmtsp}}]{\textnormal{\textsc{Multi-Depot Multiple TSP}}}
\acrodef{CSF}[\textnormal{\textsc{csf}}]{constrained spanning forest}
\acrodef{FPT}[\textsc{fpt}]{fixed-parameter tractable}
\acrodef{CPP}[\textsc{cpp}]{\textsc{Chinese Postperson Problem}}
\acrodef{RPP}[\textnormal{\textsc{rpp}}]{\textsc{Rural Postperson Problem}}
\acrodef{DRPP}[\textnormal{\textsc{drpp}}]{\textnormal{\textsc{Depot Rural Postperson Problem}}}
\acrodef{LP}[\textsc{lp}]{linear program}

\begin{document}
\title{A \texorpdfstring{$\bm{(3/2 + \varepsilon)}$}{(3/2+ε)}-Approximation for Multiple TSP with a Variable Number of Depots}


\author{
	Max Deppert \thanks{{Hamburg University of Technology, Institute for Algorithms and Complexity}, {Germany}. \textbf{email:} \texttt{max.deppert@tuhh.de}} \and
	Matthias Kaul \thanks{{Hamburg University of Technology, Institute for Algorithms and Complexity}, {Germany}. \textbf{email:} \texttt{matthias.kaul@tuhh.de}} \and
	Matthias Mnich \thanks{{Hamburg University of Technology, Institute for Algorithms and Complexity}, {Germany}. \textbf{email:} \texttt{matthias.mnich@tuhh.de}}
	}

\begin{titlepage}

\maketitle

\begin{abstract}
  One of the most studied extensions of the famous Traveling Salesperson Problem (TSP) is the {\sc Multiple TSP}: a set of $m\geq 1$ salespersons collectively traverses a set of $n$ cities by $m$ non-trivial tours, to minimize the total length of their tours.
  This problem can also be considered to be a variant of {\sc Uncapacitated Vehicle Routing}, where the objective is to minimize the sum of all tour lengths.
  When all $m$ tours start from and end at a \emph{single} common \emph{depot} $v_0$, then the metric {\sc Multiple TSP} can be approximated equally well as the standard metric TSP, as shown by Frieze (1983).
   
   
  The metric {\sc Multiple TSP} becomes significantly harder to approximate when there is a \emph{set}~$D$ of $d \geq 1$ depots that form the starting and end points of the $m$ tours.
  For this case, only a $(2-1/d)$-approximation in polynomial time is known, as well as a $3/2$-approximation for \emph{constant} $d$ which requires a prohibitive run time of $n^{\Theta(d)}$ (Xu and Rodrigues, \emph{INFORMS J. Comput.}, 2015).
  A recent work of Traub, Vygen and Zenklusen (STOC 2020) gives another approximation algorithm for metric {\sc Multiple TSP} with run time $n^{\Theta(d)}$, which reduces the problem to approximating metric TSP.
   
  In this paper we overcome the $n^{\Theta(d)}$ time barrier: we give the first efficient approximation algorithm for {\sc Multiple TSP} with a \emph{variable} number $d$ of depots that yields a better-than-2 approximation.
  Our algorithm runs in time $(1/\varepsilon)^{\mathcal O(d\log d)}\cdot n^{\mathcal O(1)}$, and produces a $(3/2+\varepsilon)$-approximation with constant probability.
  For the graphic case, we obtain a deterministic $3/2$-approximation in time $2^d\cdot n^{\mathcal O(1)}$.
\end{abstract}
\end{titlepage}

\section{Introduction}
\label{sec:introduction}
The \ac{TSP} is one of the best-studied problems in combinatorial optimization: given a complete graph $G$ on $n$ nodes together with edge weights $w: E(G)\rightarrow \mathbb R_{\geq 0}$, we seek a tour that starts at some node $v_0\in V(G)$, then visits all other nodes of $G$ exactly once, and returns to the origin $v_0$ in such a way that the overall tour weight is minimized, which is the sum of the weights of the edges traversed by the tour.
TSP is one of Karp's 21 $\mathsf{NP}$-complete problems~\cite{Karp1972}, which motivates the design of efficient, polynomial-time approximation algorithms for it.
Recall that an $\alpha$-approximation for a minimization problem returns, for any instance $I$, in polynomial time a solution of value at most $\alpha\cdot \OPT(I)$, where $\OPT(I)$ denotes the value of an optimal solution for $I$.
Of special importance in this regard is {\sc Metric TSP}, when the edge weight function $w$ obeys the triangle inequality.
For {\sc Metric TSP}, the tree doubling heuristic yields a 2-approximation, which was improved to a $3/2$-approximation by Christofides~\cite{Christofides1976} and Serdyukov~\cite{Serdyukov1978} in the 1970s.
This approximation factor stood unchallenged for many decades until its recent improvement to a $(3/2 - 10^{-36})$-approximation by Karlin et al.~\cite{KarlinKG2021}.

Due to its  ubiquity, a large variety of extensions of the \ac{TSP} have been studied.
Among the most prominent ones is the {\sc Multiple TSP}, where a set of $m\geq 1$ salespersons (all starting from some common node $v_0$ called a \emph{depot}) jointly traverse the entire set of $n$ nodes, in order to minimize the overall tour length.
That is, the goal is to find a collection of $m$ pairwise edge-disjoint cycles $C_1,\hdots,C_m$ (all intersecting in some node $v_0$) in~$G$ whose union covers all nodes of the graph and such that the sum of the weights of the cycles is minimized. 
This character of having to solve both a partitioning and a sequencing problem simultaneously gives rise to considerable added complexity, akin to that encountered in vehicle routing problems. 
Indeed, one could interpret this problem as a variant of the {\sc Uncapacitated Vehicle Routing Problem}; we, however, will adhere to the \ac{TSP}-style naming convention, since this is more prevalent in the literature.
Let us just mention here that for metric edge weights, {\sc Multiple TSP} has the same approximation guarantee as the standard (single-person) metric TSP; in particular, {\sc Metric Multiple TSP} admits a $(3/2 - 10^{-36})$-approximation in polynomial time by the results of Karlin et al.~\cite{KarlinKG2021}.
Frieze~\cite{Frieze1983} analysed the case of {\sc Metric Multiple TSP} when each of the tours has to contain at least one edge; he provided a $3/2$-approximation for this setting in polynomial time.
The {\sc Multiple TSP} is studied in more than 1,300 publications; an extensive survey is provided by Bekta\c{s}~\cite{Bektas2006}.

In this paper we study an extension of the {\sc Multiple TSP}, where a \emph{set} $D\subseteq V(G)$ of nodes is distinguished as \emph{depots}.
Formally, the \ac{MDMTSP} takes as input a complete graph $G$ on $n$ nodes together with metric edge weights $w: E(G)\rightarrow \mathbb R_{\geq 0}$, as well as a set $D\subseteq V(G)$ of $d=\abs{D}$ depots and an integer $m\geq 1$ denoting the number of salespersons available.
Now again we are seeking a set of~$m$ pairwise edge-disjoint cycles $C_1,\hdots,C_m$ in $G$ whose union covers all nodes of the graph and such that the sum of the weights of the cycles is minimized, but in addition each cycle must contain some depot from~$D$.
Such set of cycles is an optimal solution for the \ac{MDMTSP} instance, and we denote the value of some optimal solution by $\OPT(G,D,w)$ (or simply $\OPT$ if the instance is clear from the context).
The \ac{MDMTSP} is motivated by several applications of high practical impact, like motion planning of a set of unmanned aerial vehicles \cite{MalikSD2007,RathinamSD2007,YadlapalliMDP2009} and the routing of service technicians where the technicians are leaving from multiple depots~\cite{Parragh2010}.

The theoretical aspects of \ac{MDMTSP} have been studied in many research papers~\cite{Bektas2006,Bektas2012,BenaventM2013,BurgerSDS2018,KaraBektas2006,LaporteNT1988,SundarR2015,TraubVZ2020,XuR2015,XuR2017,XuXR2011}.
At this point, let us issue a word of caution.
There are quite a few other varieties of (\textsc{mdm})\ac{TSP} considered in the literature, all subtly different from each other.
For a compact overview of possible variations, there is the review paper of Bekta\c{s}~\cite{Bektas2006}.
For the scope of this paper, we consider the metric \ac{MDMTSP} where the edge weights form a metric.
This allows us to assume that $m=d$ throughout.
This assumption is made for the following two reasons: on the one hand, the case $m>d$ is negligible as the objective function (the total weight of all tours) is invariant for multiple tours starting from a single depot (if weights satisfy the triangle inequality, it is easy to show that there is always an optimal solution in which at most one route will start and end at each depot).
On the other hand, in the case $m \leq d$ we can try each selection of $d'=m$ depots by paying a multiplicative factor of $\binom{d}{m}$ in the run time only.
Thus, any instance of metric \ac{MDMTSP} is specified by a triple $(G,D,w)$, where $G$ is a complete graph on $n$ nodes, $D\subseteq V(G)$ is the set of depots, and $w:E(G)\rightarrow\mathbb R_{\geq 0}$ is a metric.
 
\emph{The polynomial-time approximability of metric \ac{MDMTSP} is not fully understood.}
That there is a \emph{set} $D$ of depots (and not just a single depot $v_0$), each one of which must be visited by one of the tours, makes the approximability of the problem much harder compared to metric {\sc Multiple TSP} (i.e., the version without depots).
The added complexity arises from the fact that we not only have to give a good order in which to visit nodes, as in the \ac{TSP}, but we also have to partition the nodes appropriately.
In particular, the Christofides-Serdyukov algorithm \cite{Christofides1976, Serdyukov1978} no longer yields a $3/2$-approximation in this setting.
The original analysis of Christofides' and Serdyukov's algorithms relies on all odd-degree nodes of some spanning structure $F$ lying on the same tour, so a parity-correcting edge set $J$ can be computed that weighs at most $\frac{1}{2}\OPT$.
This fact is not available in the multi-depot setting, so in polynomial time we can only guarantee a $2$-approximation by using the spanner $F$ for~$J$ also. 
However, this only achieves a tight approximation ratio of $2-\frac{1}{d}$ for the multi-depot setting, as shown by Xu et al. \cite{XuXR2011} (see \cref{fig:CounterexampleChristofides} for a version of their lower-bound example), because the matching can have weight $\frac{d-1}{d}\OPT$. 
\begin{figure}[htpb]
  \centering
  \includegraphics[page=1, scale = 0.7]{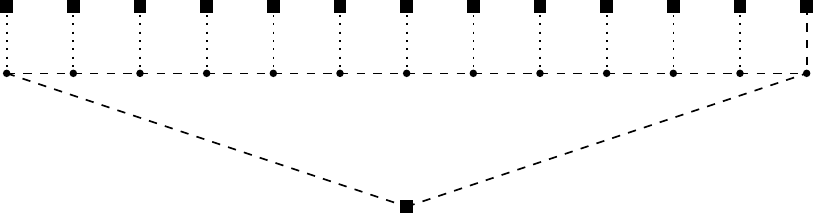}
  \caption{An instance on which \cref{alg:MDChristofides} achieves approximation ratio arbitrarily close to $2$.
    Square nodes are depots and all edges have unit weight.
    Dashed edges indicate a tour of length $d$, dotted edges a minimum \acs{CSF} with weight $d-1$ that requires a join of weight $d-1$ to complete.\label{fig:CounterexampleChristofides}}
\end{figure}
To avoid this issue, the constrained spanning forest needs to be rearranged such that there is again a matching of weight $\frac{1}{2} \OPT$, as in the work of Xu and Rodrigues \cite{XuR2015}---this rearrangement though requires $n^{\Theta(d)}$ time.

Similarly, the algorithmic approaches to metric \ac{TSP} based on solving a \ac{LP} are also unlikely to give $\alpha$-approximation algorithms with $\alpha < 2$ for metric \ac{MDMTSP}.
To this end, consider the following multi-depot version of the subtour-elimination \ac{LP}, \ref{lp:subtourElimination}:
\begin{equation*}
\label{lp:subtourElimination}\tag{MDMTSP-LP}
  \begin{array}{ll@{}ll}
    \text{minimize}  & \displaystyle\sum\limits_{e\in E(G)} w_{e}&x_{e} &\\
    \text{subject to}& \displaystyle\sum\limits_{e\in \delta(v)}   &x_{e} = 2,  &\forall \; v \in V(G)\setminus D\\
                      & \displaystyle\sum\limits_{e\in \delta(U)}   &x_{e} \geq 2,  &\forall \; U \subseteq V(G)\setminus D\\
                           &                                                &x_{e} \in [0,2], & \forall e \in E(G)
 \end{array}
\end{equation*}

The following construction in \cref{fig:CounterexampleLP} shows that \ref{lp:subtourElimination} has integrality  gap $2$:
\begin{figure}[htpb]
  \centering
  \includegraphics[page=2, scale=0.7]{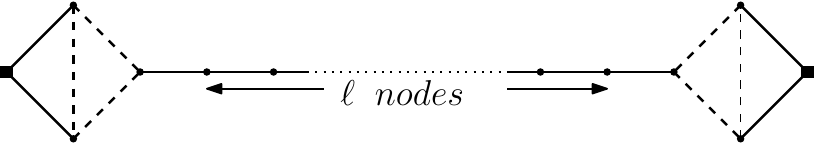}
  \caption{An instance on which the multi-depot subtour-elimination \ac{LP} has integrality gap arbitrarily close to $2$. The square nodes are the depots and all edges have unit weight. The dashed edges are assigned $x_e = 0.5$ by the \ac{LP}, the other edges $x_e = 1$. The \ac{LP} then has optimum value at most $\ell+6$, whereas the \ac{MDMTSP} has optimum value $2(\ell+4)$.\label{fig:CounterexampleLP}}
\end{figure}

If one gives up on the polynomial run time of the approximation algorithm, then smaller approximation factors are possible.
Xu and Rodrigues~\cite{XuR2015} show how to obtain a $3/2$-approximation, but their algorithm requires time $n^{\Theta(d)}$, which is polynomial only if the number $d=|D|$ of depots is constant.
Another $3/2$-approximation for \ac{MDMTSP} with run time $n^{\Theta(d)}$ follows from the recent work of Traub, Vygen and Zenklusen~\cite{TraubVZ2020}.
They in fact show the much stronger result that any $\lambda$-approximation algorithm for metric \ac{TSP} also gives a $(\lambda + \varepsilon)$-approximation algorithm for metric \ac{MDMTSP} with an additional run time factor of~$n^{\ocal(d / \varepsilon)}$.
In summary, the state-of-the-art for \ac{MDMTSP} is that there is no $\alpha$-approximation known for \ac{MDMTSP} for any absolute constant $\alpha < 2$ which runs in time $n^{o(d)}$.

\subsection{Our Results}
\label{sec:ourresults}
Our main result is a novel approximation algorithm for metric {\sc Multiple TSP} on $d$ depots with a significantly improved run time.
That is, we provide the first algorithm for metric \ac{MDMTSP} which breaks the $n^{\Theta(d)}$ time barrier to obtain an approximation ratio strictly better than~2.
Given an instance $(G,D,w)$ of metric \ac{MDMTSP}, we say a collection $C_1,\hdots,C_d$ of cycles is a \emph{tour} if they jointly cover all nodes of $G$ and each cycle contains exactly one depot from $D$.
\begin{theorem}
\label{thm:mainResult}
  There is an algorithm that, given any $\varepsilon > 0$, in time $(1/\varepsilon)^{\mathcal O(d\log d)}\cdot n^{\mathcal O(1)}$ computes a tour $T$ for any set of $n$ cities with metric distances and $d$ depots.
  The algorithm is randomized, and with constant probability the length of the tour $T$ is at most $(3/2 + \varepsilon)\cdot \OPT$.
\end{theorem}
Thus, our result significantly improves on the previously best run time $n^{\Theta(d)}$ by Xu and Rodrigues \cite{XuR2015} at the cost of some small additive $\eps$ in the approximation factor.

To break through the barrier of 2 on the approximation ratio, we need to rework the initial spanner $F$ to be \enquote{correctly aligned} with the optimal solution so that each subtour contains an even number of odd-degree nodes, as initially proposed by Xu and Rodrigues~\cite{XuR2015}.
We show that an approximate reworking can be done in time $f(d,\eps)\cdot n^{\ocal(1)}$ for some suitable function~$f$, resulting in a $(3/2+\eps)$-approximation.
To this end, firstly, we give a reduction of metric \ac{MDMTSP} to a related routing problem which is known as the \ac{RPP}.
In the \ac{RPP}, we are given an edge-weighted graph $G$ and a set $R$ of required edges, and are asked to compute a minimum-weight edge set $F$ such that $R \dcup F$ is connected and Eulerian.
Our reduction reveals an approximation algorithm with run time $\ocal(n^3+t)$ to compute solutions no worse than $\frac32\OPT+\frac12w(T)$, where $w(T)$ is the weight of a single-person \ac{TSP} tour $T$ through the depots and $t$ denotes the time to compute~$T$.
Then we use a randomized algorithm of Gutin~et~al.~\cite{GutinWY2017} for the \ac{RPP}, and an approximate weight reduction scheme of van~Bevern~et~al.~\cite{vanBevernFT2020}, to construct a $(1+\eps)$-approximation algorithm for a variant of \ac{RPP} with depots.

We are then in a position to speed up the reworking of the inital spanner due to two key insights.
Firstly, we allow for some misalignment to remain, as long as it is only due to the presence of some \emph{light} edges, limiting the number of edges we have to consider for removal from $F$.
Secondly, we employ the constructed approximation algorithm for the \ac{RPP} to complete our now disconnected spanner to a tour. 
Doing this provides a large speedup over the algorithm of Xu and Rodrigues, who first need to guess a set of edges to reconnect their spanner, and then employ a matching algorithm to obtain a tour.
Using the \ac{RPP} allows us to do both of these steps simultaneously, and considerably faster.

An important special case of metric \ac{TSP} is when the metric $w$ is induced by the shortest paths in a graph.
This version is also known as {\sc Graphic TSP}, and has been studied extensively from the perspective of approximation algorithms~\cite{MomkeS2011, SeboV2014}.
For \ac{MDMTSP} on graphic metrics, we obtain a \emph{deterministic algorithm} with slightly better approximation factor and a reduced run time.
\begin{theorem}
\label{thm:mainResultGraphic}
  There is an algorithm that, given any graph $G$ on $n$ nodes and set $D\subseteq V(G)$ of $d$ depots, in time $2^{d}\cdot n^{\mathcal O(1)}$ computes a tour $T$ of length at most $\frac{3}{2}\cdot \OPT$.
\end{theorem}

\subsection{Related work}
\label{sec:relatedwork}
In the case of a single salesperson, i.\,e. $m=1$, B{\'e}rczi et al. \cite{BercziMV2022} gave a polynomial-time $3/2$-approximation for the \emph{many-visits} version of metric \ac{MDMTSP}, that is, when each node $v$ is equipped with a request $r(v)$ (encoded in binary) of how many times it should be visited.

In a different work, B{\'e}rczi et al. \cite{BercziMV2023} have shown constant-factor approximation algorithms with ratio at most $4$ for variants of metric \ac{MDMTSP} where each tour has to visit exactly one depot (and thus, $d \leq m$).

For the \ac{RPP}, which asks for a minimum-weight tour traversing all edges of a given subset $R$ of edges of a graph, there is a polynomial-time approximation algorithm (cf. Frederickson~\cite{Frederickson1979} or Jansen~\cite{Jansen1992}) similar to the approach of Christofides-Serdyukov for metric \ac{TSP}.
The weight of a solution can be bounded by $(3\OPT+w(R))/2$, where $w(R)$ is the weight of $R$.

Oberlin et al.~\cite{OberlinRD2009} studied heuristic approaches for \ac{MDMTSP} where $d=m$ and each salesperson is located at its own depot.

For the objective of minimizing the longest tour length of any salesperson (rather than the sum of all the tour lengths), Frederickson et al. \cite{FredericksonHK1978}, among other routing problems, considered the case of a single depot ($d=1$), and presented a $(\rho + 1 - 1/m)$-approximation algorithm where $\rho$ is the approximation ratio of an algorithm for the single-salesperson \ac{TSP}.

\section{Preliminaries}
\label{sec:preliminaries}
Let $\mathcal U$ be a finite universe.
For a function $w : \mathcal{U} \to \mathbb{R}$ and a multiset $U \subseteq \mathcal{U}$ we write $w(U)$ to mean $\sum_{u\in U} w(u)$, where the sum has an additional summand for each copy of an element in~$U$, i.e. it considers multiplicities.
The disjoint union $\dot \bigcup_i A_i$ of some sets $\{A_i\}_i$ is considered to be the multiset of all items in the collection.
For brevity we often write $2A$ to mean $A \dcup A$.

Throughout this paper, we consider the multiple-depot version of the metric {\sc Multiple TSP}, or metric \ac{MDMTSP} for short.
We will generally represent the metric by an edge-weighted graph whose shortest-path metric we assume to be the metric in use.
Notice that this makes no difference in our setting since we are allowed to traverse edges multiple times; only if this is forbidden does the non-metric case become relevant.

\begin{definition}
An instance $(G,D,w)$ of metric \ac{MDMTSP} consists of a complete graph~$G$, a set $D \subseteq V(G)$ of $d$ \emph{depots}, and a metric $w : E(G) \to \mathbb{R}_+$ on $V(G)$. 
A multiset of edges $T \subseteq E(G)$ is called a \emph{tour} of $(G,D,w)$ if\\
\vspace{0.1em}\\
\setlength\parindent{1em}\indent
\begin{minipage}{\textwidth-1em}
\begin{enumerate}
	\item[{\crtcrossreflabel{(P1)}[Prop:parity]}] the multigraph $(V(G),T)$ has even degree at every node in $V(G)$,
	\item[{\crtcrossreflabel{(P2)}[Prop:connectivity]}] and each connected component of $(V(G),T)$ contains at least one node from $D$.
\end{enumerate}
\end{minipage}\\
\vspace{0.1em}\\
We denote by $\OPT(G,D,w)$ the minimum weight of any tour of $(G,D,w)$. 
If the instance is clear from the context, we may only say $\OPT$.
\end{definition}
Edge sets are generally allowed to be multisets, and graphs can have parallel edges. 

Imitating the general framework of Christofides-Serdyukov \cite{Christofides1976, Serdyukov1978}, we first compute an edge set $F$, called a \ac{CSF}, that ensures the connectivity property~\ref{Prop:connectivity}. 
We then compute an additional set of edges $J$ such that $F\dcup J$ has property~\ref{Prop:parity}.
\begin{definition}
	Let $G$ be a graph and let $D \subseteq V(G)$.
  A \emph{constrained spanning forest} in $(G,D)$ is a set $F \subseteq E(G)$ of edges such that the graph $(V(G),F)$ is acyclic and every connected component of $(V(G), F)$ contains at least one node from $D$.
\end{definition}
We will make use of the following result.
\begin{theorem}[Rathinam et al. \cite{RathinamSD2007}]
  Given any graph $G$ on $n$ nodes and $m$ edges with weights $w:E(G)\rightarrow\mathbb R$ and a set $D\subseteq V(G)$, a minimum-weight \ac{CSF} of $(G,D,w)$ can be computed in time $\ocal((n+m)\log n)$.
\end{theorem}
\begin{proof}
	The computation of a minimum-weight \ac{CSF} for $(G,D,w)$ can be reduced to computing a minimum-weight spanning tree in a graph $G'$ with edge weights $w'$.
  The graph $G'$ is obtained from $G$ by adding a single root node $r$ connected to every depot by an edge of some weight less than the weight all other edges in $G$.
	Kruskal's algorithm is guaranteed to choose these edges for the minimum-weight spanning tree in $(G',w')$, and after removing $r$ we are left with a minimum-weight \ac{CSF} for $(G,D,w)$.
\end{proof}
Traditionally, property \ref{Prop:parity} is obtained by computing a minimum-weight matching on the odd-degree nodes of the \ac{CSF}, as in \cref{alg:MDChristofides}.
\begin{algorithm}[ht]
	\SetAlgoLined
	\KwIn{A metric \ac{MDMTSP} instance $(G, D, w)$.}
	\KwOut{A tour $T$ with $w(T) \leq 2\OPT$.}
	Compute a minimum-weight \ac{CSF} $F$ for $(G,D,w)$\;
	Let $U$ be the set of nodes with odd degree in $F$\;
	Compute a minimum-weight perfect matching $M$ in $G[U]$\;
	\KwResult{$T := F \dcup M$}
	\caption{Algorithm \textsc{Multi-Depot Christofides-Serdyukov}\label[algorithm]{alg:MDChristofides}}
\end{algorithm}
However, this algorithm only achieves a tight approximation ratio of $2-\frac{1}{d}$ for the multi-depot setting, as shown by Xu et al. \cite{XuXR2011} (see \cref{fig:CounterexampleChristofides} for a simplified version of their lower bound example), because the matching can have weight $\frac{d-1}{d}\OPT$. 
To avoid needing such an expensive matching, the constrained spanning forest needs to be rearranged such that there is again a matching of weight $\frac{1}{2} \OPT$, as in the work of Xu and Rodrigues \cite{XuR2015}.

%

\section{Reducing Multi-Depot Multiple TSP to Rural Postperson Problem}
\label{sec:reduction}
In this section we show a reduction from the metric \ac{MDMTSP} to the \ac{RPP}.
Recall that in the \ac{RPP} there is a required set $R$ of edges that a tour should traverse, rather than a set of nodes. 
\begin{definition}[Rural Postperson Problem]
  An instance $(G,R,w)$ of \ac{RPP} consists of a graph~$G$, a set $R \subseteq E(G)$ of required edges, and a metric weight function $w : E(G) \to \mathbb{R}_{\geq 0}$.
  A \emph{solution} is multiset $J \subseteq E(G)$ for which $(V, R\dcup J)$ is Eulerian, and which has only one non-singleton connected component.\footnote{This means that nodes not incident to any edge from $R$ do not need to be visited by the computed tour. For metric cost functions, however, one can always reduce to the case where $R$ spans $G$.}
  The weight of a solution $J$ is $w(J) = \sum_{e\in J}w(e)$.
  The goal of \ac{RPP} is to compute an \emph{optimal} solution, which is a solution of minimum weight~$\OPT(G,R,w)$.
\end{definition}

There is a polynomial-time approximation algorithm for \ac{RPP} \cite{Frederickson1979, Jansen1992} which computes a solution $J \subseteq E(G)$ such that $w(R\dcup J) \leq \frac{3}{2} w(R\dcup J^*)$, i.\,e. $w(J) \leq \frac32\OPT+\frac12w(R)$, where~$J^*$ is some optimal solution.
Due to the first inequality and the unavoidable weight of~$R$, the algorithm is known as a $3/2$-approximation for \ac{RPP}.
This situation is very similar to the current approximation status of metric \ac{MDMTSP}, where we can obtain a $3/2$-approximation if we allow for some additional additive term.
This observation motivates the following reduction from metric \ac{MDMTSP} to \ac{RPP}.
\begin{observation}
\label{obs:MDMTSPtoRPP}
  For each instance $(G,D,w)$ of \ac{MDMTSP} there is an instance $(G', R, w')$ of \ac{RPP} such that any solution to the \ac{RPP} instance can be transformed in polynomial time into a solution to the \ac{MDMTSP} instance of the same weight. 
\end{observation}
\begin{proof}
  First, compute any \ac{TSP} tour $S$ on the depots in $G$, that is, on $G[D]$.
  Then, for each node $v\in V\setminus D$, introduce a second node $v'$, as well as an edge $e_v = \{v,v'\}$, and set its weight to $w'(e_v) = 0$.
  For each edge $e\in E(G)$ set $w'(e) = w(e)$, and set $R$ to be the union of $S$ and two copies of each $e_v$.
  The any solution to the constructed instance of \ac{RPP} corresponds to an \ac{MDMTSP} tour for $(G,D,w)$ of the same weight.
\end{proof}
Notice that this reduction, together with the $3/2$-approximation for \ac{RPP}, allows us to compute a solution to \ac{MDMTSP} of weight at most $\frac{3}{2}\OPT + \frac{1}{2}w(S)$.
In particular, if all depots are pairwise close to each other this is already a better-than-2 approximation.
    
In \cref{sec:formalAlgorithmDescription} we will in some sense show a stronger result that there is also a (Turing) reduction from \ac{MDMTSP} to the special case of \ac{RPP} where $(V,R)$ has few connected components, which has been shown by Gutin et al. to be tractable \cite{GutinWY2017}:
\begin{proposition}[Gutin et al. \cite{GutinWY2017}]
\label{thm:k-RPP}
  There is a randomized algorithm for \ac{RPP} that for any instance $(G,R,w)$, where $(V(G),R)$ has $k$ connected components and $w$ takes only integer values, in time $2^{\ocal(k)}(n + \OPT(G,R,w))^{\ocal(1)}$ produces a solution.
  With constant probability, the computed solution is optimal.
\end{proposition}

However, as our reduction is only $(1+\eps)$-approximate with respect to the solution qualities, and needs time exponential in $d$ and $\eps$, we need to remove the polynomial dependence on $\OPT(G,R,w)$ in \cref{thm:k-RPP}.
To this end, we will adapt an approximate weight reduction scheme by van Bevern et al. \cite{vanBevernFT2020}:

\begin{lemma}[{adapted from van Bevern et al. \cite[Lemma 2.12]{vanBevernFT2020}}]
\label{lem:weightReduction}
  Let $(G,R,w)$ be an instance of \ac{RPP} with integral weighs, let $\varepsilon>0$, and let $\beta = \max\{w(e) \mid e\in E(G)\}$.
  Then in polynomial time we can compute a weight function $w': E(G) \to \mathbb{N}_{\geq 0}$ such that
  \begin{itemize}
    \item $\max\{w'(e) \mid e \in E(G)\} \leq 2\abs{E(G)} / \eps$,
    \item and for all $\alpha \geq 1$, any solution $J$ to $(G,R,w')$ with weight $w'(J) \leq \alpha \OPT(G,R,w')$ also fulfills $w(J)~\leq~\alpha\OPT(G,R,w)~+~\eps \beta$, as long as $J$ contains at most two copies of each edge.
   \end{itemize}
\end{lemma}
\begin{proof}
	The rounding scheme simply sets $w'(e) := \lfloor w(e)\cdot \frac{2|E(G)|}{\varepsilon \cdot \beta} \rfloor$ for each edge $e$ of $G$.
	This yields the first condition, by definition. 
	For the second condition, observe that for any~$J$ we have 
	\begin{equation*}
		\frac{\varepsilon \cdot \beta}{2|E(G)|} w'(J) \leq w(J) \leq \frac{\varepsilon \cdot \beta}{2|E(G)|} w'(J) + \frac{\varepsilon \cdot \beta}{2|E(G)|}|J| \leq \frac{\varepsilon \cdot \beta}{2|E(G)|} w'(J) + \varepsilon \beta \enspace .
	\end{equation*}
	Hence, the two weight functions are equivalent up to scaling by a constant and the addition of at most $\varepsilon \beta$.
\end{proof}
Notice that the restriction on $J$ having at most two copies of each edge is never a problem: whenever a solution to the \ac{RPP} has three or more copies of one edge, we can delete two of them to obtain a cheaper solution.

We will combine \cref{thm:k-RPP} and \cref{lem:weightReduction} to obtain an approximation for $k$-component \ac{RPP} whose run time does not depend on $\OPT$.
\begin{corollary}
\label{cor:apxk-RPP}
  There is a randomized algorithm that, for any $\eps > 0$, in time $2^{\ocal(k)}(n + \frac{1}{\eps})^{\ocal(1)}$ computes a solution~$J$ for any instance $(G,R,w)$ of \ac{RPP} where $(V(G),R)$ has $k$ connected components.
  The computed solution $J$ has the property that, with constant probability, $w(J) \leq \OPT(G,R,w) + \eps \max\{ w(e) \mid e \in J^* \dcup R\}$, where~$J^*$ is some optimal solution to the instance.
\end{corollary}
\begin{proof}
  We first guess the weight $\beta$ of the most expensive edge in $J^* \dcup R$, where $J^*$ is some optimal solution.
  There are only $\abs{E(G)}$ options, so the guessing generates only polynomial overhead. 
  All edges that are more expensive than $\beta$ can be removed from the instance to get some graph $G'$.
  Now we apply \cref{lem:weightReduction} to get an instance with weights $w'$ bounded by $2\abs{E(G)}/\eps$ and use the exact algorithm from \cref{thm:k-RPP} to get a solution $J$ to $(G',R,w')$ in time $2^{\ocal(k)}(n + \frac{1}{\eps})^{\ocal(1)}$.
  From \cref{lem:weightReduction} with $\alpha = 1$ we know that
  \begin{equation*}
    w(J) \leq \OPT(G,R,w) + \eps \beta \leq \OPT(G,R,w) + \eps \max\{ w(e) \mid e \in J^* \dcup R\},
  \end{equation*}
  which proves the claim.
\end{proof}
    
We will be using this algorithm to complete partial solutions to instances of \ac{MDMTSP}.
We will need only a slight modification that allows for the presence of depots as follows.
\begin{definition}[Depot Rural Postperson Problem]
  An instance $(G,D,R,w)$ of the \ac{DRPP} consists of an \ac{RPP} instance $(G,R,w)$ and some depots $D \subseteq V(G)$.
  A \emph{solution} is a multiset $J \subseteq E(G)$  such that $(V, R\dcup J)$ is Eulerian and each non-singleton connected component of $(V,R\dcup J)$ contains at least one depot.
  The weight of a solution $J$ is $w(J) = \sum_{e\in J}w(e)$.
  The goal is to compute an \emph{optimal} solution, which is a solution of minimum weight $\OPT(G,D,R,w)$.
\end{definition}

The depot version \ac{DRPP} can be reduced to regular \ac{RPP} quite easily.
\begin{corollary}
\label{cor:apxk-DRPP}
  There is a randomized algorithm that, for any instance $(G,D,R,w)$ of \ac{DRPP} where $(V(G),R)$ has $k$ connected components and any $\eps > 0$, in time $2^{\ocal(k\log k)}(n + \frac{1}{\eps})^{\ocal(1)}$ computes a solution~$J$ such that, with constant probability, $w(J) \leq (1+\eps)\OPT(G,D,R,w) + \eps w(R)$.
\end{corollary}
\begin{proof}
  Note first that each connected component of $(V(G), R)$ can be assumed to contain at most one depot, so $\abs{D} \leq k$.
  Some optimum solution $J^*$ induces a partition of the connected components of $(v(G),R)$ where each partition class corresponds to those components connected to some specific depot.
  There are at most $\abs{D}^k \in 2^{\mathcal O(k\log k)}$ possible partitions, so we can try each partition, solve the regular \ac{RPP} instance on each of the $\abs{D}$ classes of the partition using the algorithm from \cref{cor:apxk-RPP}, and return the best solution we found.
\end{proof}

\section{Intuition for the Algorithm}
The algorithm of Xu and Rodrigues~\cite{XuR2015} executes, at a very high level, the following steps:
\begin{enumerate}
	\item Compute a minimum-weight constrained spanning forest $F$ for $(G,D,w)$.
	\item Guess a set $X$ of at most $|D|-1$ edges such that they are in $F$ but not in some fixed optimal tour $T$.
	\item Discard the guessed edges $X$ from $F$.
    This leaves at most $2\abs{D}$ connected components in $(V,F\setminus X)$.
    If we have guessed correctly, every subtour of $T$ now contains an even number of odd-degree nodes.
    There must exist some edges $A$ from $T$ such that $(F\setminus X)\cup A$ is a \ac{CSF} for $(G,D)$ with $w((F\setminus X)\cup A) \leq w(T)$.
    The value $\abs{A}$ is at most $\abs{D}$, so we also guess $A$.
	\item Since $A$ contains only edges from $T$, every subtour of $T$ still contains an even number of odd-degree nodes with respect  to $(V,(F\setminus X)\cup A)$.
    If we compute an odd-join $J$ for $(F\setminus X)\cup A$, we have $w(J) \leq \frac{1}{2}\OPT$, so return $((F\setminus X)\cup A) \dcup J$.
\end{enumerate}
Since the algorithm needs to guess $2\abs{D}$ edges in total (in step 2), it can be implemented in time $n^{\mathcal O(d)}$.
We modify this guessing step by considering for discarding (in step 3) only very heavy edges, and by sidestepping the guessing of $A$; instead of computing first a connected structure and then a join we do this simultaneously, using the algorithm for \ac{RPP}.
Specifically:
\begin{itemize}
	\item In step 2, we only consider edges that are very expensive relative to the total weight of the forest $F$.
    If the targeted edge $e$ is not in this collection, we do not delete it but instead use it as part of the augmenting set $A$, doubling the edge.
    This also fixes parity, but requires us to relax $w((F\setminus X)\dcup A) \leq w(T)$ to $w((F\setminus X)\dcup A) \leq (1+\eps)w(T)$.
    The $\eps$ can be controlled by how expensive relative to $F$ we allow these non-deleted edges to be.
	\item In step 3, we do not actually guess $A$, we merely use its existence.
    We instead solve an instance of \ac{DRPP} with at most $2d$ connected components for which $A\dcup J$ is a solution.
    Using the algorithm from \cref{cor:apxk-DRPP}, we can compute a $(1+\eps)$-approximation for the \ac{DRPP} in time $f(d,\varepsilon)\cdot n^{\mathcal O(1)}$.
    We use the solution  $J'$ as a replacement for $A\dcup J$ knowing $w(J') \leq (1+\eps) w(A\dcup J)$.
    Combining inequalities for $J$ and $F$ gives: 
	\begin{equation*}
	  w((F\setminus X) \dcup J') \leq (1+\eps)w(((F\setminus X)\dcup A)\dcup J) \leq (1+\eps) \frac{3}{2}\OPT \enspace .
	\end{equation*}
\end{itemize}
	
An illustration of the augmentation scheme can be found in \cref{fig:Augmentation}.
\begin{figure}[htpb]
	\centering
	\includegraphics[page=1, scale = 0.7]{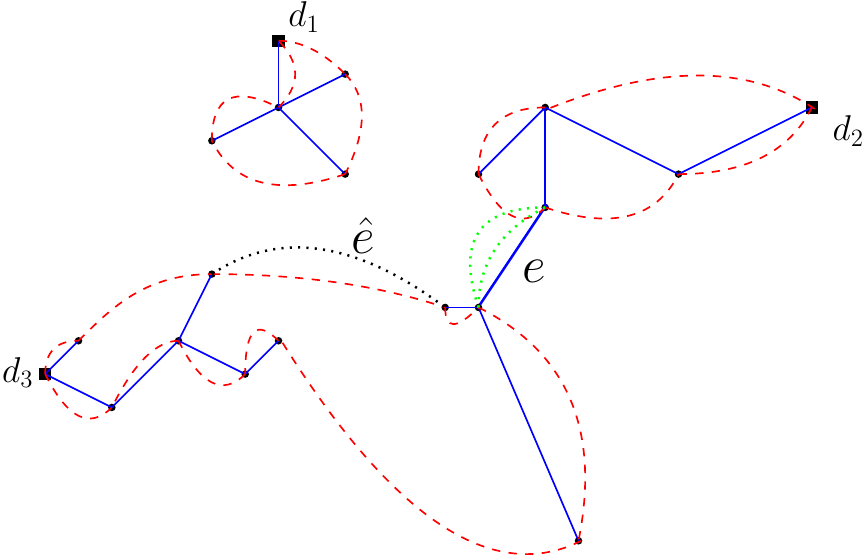}
	\caption{Illustration of the augmenting edges explored by our algorithm.
    Blue solid edges represent a \ac{CSF}, dashed red edges an optimal tour. 
		Note that the tours of $d_2$ and $d_3$ have an odd number of nodes with odd degree in the \ac{CSF}.
    Our algorithm considers two options to remedy this.
		We either add two copies (green, dotted) of the edge $e$ to the optimal tour, if $e$ is considered to be light enough. 
		This joins the tours of $d_2$ and $d_3$ to a single tour with an even number of odd-degree nodes.
		If $e$ is considered too heavy for this, we remove $e$ from the \ac{CSF} and replace it with the edge~$\hat{e}$ (black, dotted).
    As $\hat{e}$ comes from the optimal tour, this keeps the cost of the \ac{CSF} below $\OPT$ and it fixes the parities.
	}
\label{fig:Augmentation}
\end{figure}

\section{Towards Faster Parity Correction}
\label{sec:formalAlgorithmDescription}
In this section, we give a formal version of the algorithm described in the previous section, which we state as \cref{alg:1.5}. We prove this algorithm to be a $(3/2+\ocal(\eps))$-approximation for metric \ac{MDMTSP} in \cref{thm:algCorrectness}. 
We will restate and reprove some of the results of Xu and Rodrigues \cite{XuR2015} to ensure completeness of the presentation and to integrate properly our changes to their algorithm.

To this end, let us fix some notation throughout this section.
Let $(G,D,w)$ be a metric \ac{MDMTSP} instance with $D = \{d_1,\dots,d_d\}$, an optimal tour $T$, and the minimum-weight CSF $F$ for $(G,D,w)$ that was computed \cref{alg:MDChristofides}.
We denote by $T_i$ be the connected component of $T$ containing $d_i$, by $F_i$ the subtree of $F$ containing $d_i$, and $U_i$ the set of nodes in $T_i$ that have odd degree with respect to the edges in $F$.
We take $U = \bigcup_{i}U_i$.

By minimality of $F$, we already know that $w(F) \leq w(T) = \OPT$.
To extend $F$ to an \ac{MDMTSP} tour, we try to compute a minimum-weight matching between the nodes in $U$. 
It is a standard argument from the analysis of the Christofides-Serdyukov algorithm that if~$\abs{U_i}$ is even, $T_i$ contains two disjoint matchings for the nodes in $U_i$. 
So if every $U_i$ has even cardinality, then any minimum-weight matching has weight at most $\frac{1}{2} \OPT$. 
But this is not the case, since a tree $F_i$ might contain nodes from many different tours, so the odd-degree nodes are distributed arbitrarily. 
To record this ``misalignment'' between the trees and subtours we introduce the concept of an \emph{alignment graph}.
\begin{definition}[Alignment Graph]
  The alignment graph $H$ for $(G,F,T)$ is constructed as $V(H) = D$,  and
  \begin{equation*}
    E(H) = \{ \{d_i, d_j\} \mid \exists e \in F \text{ s.t. } \abs{e \cap V(T_i)} = \abs{e \cap V(T_j)} = 1\} \enspace .
  \end{equation*}
  We also define a weight function $w_H : E(H) \to \mathbb{R}_+$ as
  \begin{equation*}
    w_H( (d_i, d_j) ) := \min\{w(e) | e \in F, \;\abs{V(T_i) \cap e} = \abs{V(T_j) \cap e} = 1\} \enspace .
  \end{equation*}
\end{definition}
In the following, we assume that $H$ is connected, otherwise the analysis holds independently for each connected component.

Now we take $D_{\text{odd}}$ to be the collection of depots $d_i$ for which $\abs{U_i}$ is odd, and $A_H$ to be any $D_{\text{odd}}$-join in $H$.
The join $A_H$ can be used to augment the original tour $T$ to be connected.
To do this we transfer the join to the original graph to ensure that it contains a ``cheap'' matching. 
For every edge $e = \{d_i, d_j\} \in E(A_H)$, pick an edge in $\hat{e} \in E(F)$ with $w(\hat{e}) = w_H(e)$ and $\abs{\hat{e} \cap V(T_i)} = \abs{\hat{e} \cap V(T_j)} = 1$.
Denote by $A$ the collection of these $\hat{e}$. 
Observe that every node in  $T \dcup 2A$ has even degree, and every connected component of the graph contains an even number of nodes from $U$.
Hence, there exists a $U$-join $J$ in $G$ with $w(J) \leq \frac{1}{2}w(T\dcup 2A)$.

Now notice that, if the edges in $A$ have weight at most $\eps w(F)$, this inequality yields that \cref{alg:MDChristofides} already achieves a good approximation ratio, specifically
\begin{equation*}
  w(F) + w(M) \leq w(F) + w(J) \leq (1+\eps) w(F) + \frac{1}{2}w(T) \leq (\frac{3}{2}+\eps) w(T) \enspace .
\end{equation*}
Based on this observation, we are willing to augment $T$ with low-weight edges from $F$ to find a low-weight matching.
Therefore, we need to distinguish between \emph{heavy} and \emph{light} edges.
\begin{definition}
  Let $\eps > 0$.
  An edge $e \in F$ is $\eps$-\emph{light} if $w(e) \leq \frac{\eps}{d} \cdot w(F)$; else, it is $\eps$-\emph{heavy}.
\end{definition}

We now try to replace the $\eps$-heavy edges in $A$ with some other edges from $T$. 
\begin{lemma}[{compare \cite[Section 2]{XuR2015}}]
\label{lem:AugmentationReplacement}
  Let $X \subseteq A$.
  Then there exist a set $\hat{A} \subseteq E(T)$ of edges such that $(F\setminus X)\cup\hat{A}$ is a CSF for $(G,D)$ with $w((F\setminus X)\cup\hat{A}) \leq w(T)$.
\end{lemma}
\begin{proof}
  Consider the forest $F'$ obtained from $T$ by removing exactly one edge from each subtour. 
  $F'$ contains only edges from $T$, so it is disjoint from $X$. 
  By a standard matroid exchange argument, for each $e \in X$ there is a $\hat{e} \in F'$ such that $F-e+\hat{e}$ is a CSF and $w(e) \leq w(\hat{e})$. This process can then be iterated to remove all of $X$.
  The collection of these $\hat{e}$ is $\hat{A}$, giving $w((F\setminus X)\cup\hat{A}) \leq w(F') \leq w(T)$.
\end{proof}

This process of replacing augmenting edges from $A$ with edges out of $T$ also fulfills the key goal of putting an even number of odd-degree nodes into every connected component of some augmented \ac{MDMTSP} solution. 
Consider the following lemma, which is in substance a version of a statement by Xu and Rodrigues~\cite[Theorem 2]{XuR2015}.
\begin{lemma}
\label{lem:parityThroughRemoval}
  Let $A,X,\hat{A}$ be as in \cref{lem:AugmentationReplacement}.
  Then every connected component of $T \cup (A \setminus X)$ contains an even number of nodes that have odd degree in $(F\setminus X)\cup\hat{A}$.
\end{lemma}
\begin{proof}
  Notice first that the connected components of $T \cup (A\setminus X)$ are the union of some of the subtours of $T$. 
  Since the edges in $\hat{A}$ belong to some tour, adding them to $F\setminus X$ flips the parity of the degrees of two nodes on the same tour, so the total parity of odd-degree nodes on that tour does not change.
  We can therefore restrict ourselves to considering the odd-degree nodes with respect to $F\setminus X$.
  
  Recall that originally $A$ was constructed from a $D_{\text{odd}}$-join $A_H$ in the alignment graph~$H$.
  So~$X$ corresponds to some edge set $X_H \subseteq A_H$, and we know that $A_H \setminus X_H$ constitutes an $(D_{\text{odd}} \Delta \dot \bigcup_{e \in X_H} e )$-join, where $\Delta$ denotes the symmetric difference.
  For multisets, the symmetric difference of some sets contains an item if and only if it is contained an odd number of times in their disjoint union.
  At the same time, removing an edge $e \in X$ from $F$ with corresponding $\{d_i,d_j\} \in X_H$ changes the degree of one node in $V(T_i)$ and one node on~$V(T_j)$. 
  So, the depots whose tours contain an odd number of odd-degree nodes with respect to $F\setminus X$ are precisely $(D_{\text{odd}} \Delta \dot \bigcup_{e \in X_H} e )$, so $A_H \setminus X_H$ joins them correctly.
\end{proof}

We are now ready to prove that \cref{alg:1.5} returns a $(3/2 + \ocal(\eps))$-approximation, with constant probability.
\begin{algorithm}[t]
	\SetAlgoLined
	\KwIn{A metric \ac{MDMTSP} instance $(G, D, w)$ and a parameter $\eps>0$.}
	\KwOut{A tour $T$ such that $w(T) \leq (3/2+\eps)\OPT$ with constant probability.}
	Compute a minimum-weight CSF $F$ for $(G,D,w)$\;
  Let $T$ be the currently best \ac{MDMTSP} solution, initially $2F$\;
	Let $Y$ be the set of edges in $F$ which are $\eps$-heavy\;
	\ForEach{$X \subseteq Y$, $\abs{X} \leq \abs{D}$}{
    $F' := F \setminus X$\;
    Compute a solution $M$ for the \acs{DRPP} instance $(G,D,F',w)$ using \cref{cor:apxk-DRPP} \;
    \If{$w(M\dcup F') < w(T)$}{ set $T = M\dcup F'$\;}
  }
	\KwResult{$T$}
	\caption{Algorithm \textsc{Extended Multi-Depot Christofides-Serdyukov}\label[algorithm]{alg:1.5}}
\end{algorithm}    
\begin{theorem}
\label{thm:algCorrectness}
  The tour returned by \cref{alg:1.5} has weight at most $(3/2 + \ocal(\eps)) \OPT$, with constant probability.
\end{theorem}
\begin{proof}
  Set $T'$ to be the tour returned by the algorithm.
  Now let $A$ be the augmenting edge set for some optimal tour as before and $X$ the set of $\eps$-heavy edges in $A$. 
  We look at the iteration of the algorithm where that $X$ is considered for removal.
  From \cref{lem:AugmentationReplacement} we know that there exists some edge set $\hat{A}$ with $w(F\setminus X\cup \hat{A}) \leq \OPT$ and \cref{lem:parityThroughRemoval} implies that there is an edge set $J$ such that $F\setminus X\dcup \hat{A} \dcup J$ is Eulerian, contains a depot in each connected component, and $w(J) \leq \frac{1}{2} w(T \cup (A\setminus X))$.
  Therefore, $\hat{A} \dcup J$ is a solution to the \ac{DRPP} instance $(G,D,F',w)$. 
  Hence, the $M$ computed in the algorithm fulfills $w(M) \leq (1+\eps) w(\hat{A} \dcup J) + \eps w(F')$.
  Putting all these inequalities together yields
  \begin{align*}
    w(T') &= w(M) + w(F \setminus X) \\
          &\leq (w(F) - w(X) + w(\hat{A})) + (\frac{1}{2}(w(T) + w(A\setminus X)) + \eps(w(\hat{A}) + w(J) + w(F))\\
          &\leq w(T) + \frac{1}{2}(w(T) + d\cdot \frac{\eps}{d}w(F))+ \eps(w(\hat{A}) + w(J) + w(F))
          \leq \frac{3}{2} w(T) +4\eps w(T),
  \end{align*}
  where we use that $A \setminus X$ contains at most $d$ edges, and all of them are $\eps$-light.
\end{proof}
Notice that this algorithm will give a $(3/2+\eps)$-approximation when called with $\eps / 4$ as the parameter of approximation.
The additional run time cost will vanish in the $\ocal$-notation.
The probability of success for this algorithm is the same as that for the algorithm in \cref{thm:k-RPP}. 
Notice that while that algorithm is called many times, we only need it to succeed for one specific choice of $X$. 
If it fails in one of the other attempts, we do not care.
    
It remains to analyze the run time of this algorithm. 
We see that $Y$, the set of $\eps$-heavy edges, has size at most $\frac{d}{\eps}$, so there are only $(\frac{d}{\eps})^d$ possible values for $X$ to be tried.
Note also that each loop iteration requires the approximate solution of a \ac{DRPP} instance with $\ocal(d)$ components which can be done in time $2^{\ocal(d\log d)}(n + \frac{1}{\eps})^{\ocal(1)}$.
The total run time then is $(1/\eps)^{\ocal(d\log d )}\cdot n^{\ocal(1)}$, showing \cref{thm:mainResult}.

\section{A Deterministic \texorpdfstring{$\bm{3/2}$}{3/2}-Approximation for Graphic MDMTSP}
\label{sec:graphic-case}
In this section we provide a deterministic $3/2$-approximation for \ac{MDMTSP} when the metric is the shortest-path metric of an unweighted graph.
The run time of the algorithm is $2^d\cdot n^{\mathcal O(1)}$.

Let $(G,D)$ be an instance of graphic \ac{MDMTSP}, where $G$ this time is the unweighted graph inducing the shortest-path metric.
Note that we can assume $G$ to be connected. This allows us to construct \ac{TSP} tours that are not much more expensive than optimal solutions to \ac{MDMTSP}, which re-enables the original analysis of the Christofides-Serdyukov Algorithm.

For a given optimal \ac{MDMTSP} tour $T$, we can extend it to a \ac{TSP} tour by introducing at most $2(d-1)$ edges.
To do this, contract the subtours of $T$, find a spanning tree in the contracted graph, and double all the edges of that tree.
We then see that the solution $F\dcup M$ returned by \cref{alg:MDChristofides} fulfills
\begin{equation*}
  w(F\dcup M) \leq w(T) + \frac{1}{2}(w(T) + 2(d-1)) = \frac{3}{2}w(T) + d-1 \enspace .
\end{equation*}
Notice that the additive term $d-1$ is likely to be very small, since we know $w(T) \geq n-d$. 
A similar argument can also be made for metrics which are continuous in the sense that the space cannot be partitioned into two very distant parts.
\begin{observation}
\label{obs:contMetrics}
  Let $(G,D,w)$ be an \emph{integer-weighted} instance of \ac{MDMTSP} for which there exists a constant $L$ such that, for all $U\subseteq V(G)$, it holds
    $\min\{w(u,v) \mid u\in U, v\not \in U\} \leq L$.
  Then \cref{alg:MDChristofides} returns a solution $T$ for $(G,D,w)$ with $w(T) \leq \frac{3}{2}\OPT(G,D,w) + L(d-1)$.
\end{observation}
Since we know $w(T)$ to be in $\Omega(n-d)$ also in this case, \cref{alg:MDChristofides} gives an \emph{asymptotic} $3/2$-approximation for any constant $d$ and $L$.
We can even get rid of the additive term in the graphic case (i.e. $L=1$) with some additional run time. 
\begin{observation}
\label{obs:graphic1.5}
  There is a $3/2$-approximation algorithm for graphic \ac{MDMTSP} with run time $2^d\cdot n^{\ocal(1)}$.
\end{observation}
\begin{proof}
  Let $T$ be some fixed optimal tour.
  We start by guessing the set $D'\subseteq D$ of depots whose subtours in $T$ contain at least one edge, generating on overhead of $2^d$. 
  Then we know that $T$ contains a \ac{CSF} $F'$ for $(G,D')$ with weight $\abs{E(T)} - \abs{D'}$. 
  As before, we connect together all subtours of the depots in $D'$ with $\abs{D'}-1$ edges, and double these edges.
  Then the tour $F\dcup M$ returned by \cref{alg:MDChristofides} fulfills 
  \begin{equation*}
    \abs{E(F\dcup M)} \leq \abs{E(T)} - \abs{D'} + \frac{1}{2}(\abs{E(T)} + 2(\abs{D'}-1)) = \frac{3}{2} \abs{E(T)} - 1,
  \end{equation*}
  and that proves the claim.
\end{proof}
For the special case where we require each depot to have a non-empty tour, we do not even have to guess the correct subset of depots in \cref{obs:graphic1.5}, yielding a $3/2$-approximation in truly polynomial time.

\section{Discussion}
\label{sec:discussion}
We have shown that metric \ac{MDMTSP} admits a randomized $(3/2+\varepsilon)$-approximation algorithm in time $(1/\varepsilon)^{\ocal(d\log d)}\cdot n^{\ocal(1)}$, filling in the gap between the best-known polynomial approximation factor, $2$, and the $3/2$-approximation of Xu and Rodrigues in time $n^{\Theta(d)}$.
However, there remain a number of natural openings for improving on our result:
\begin{itemize}
  \item Can our algorithm be derandomized?
    Since we rely on the algorithm of Gutin et al. \cite{GutinWY2017} to solve \ac{RPP} instances, this would require a derandomization of their result.
    However, their algorithm relies on the Schwartz-Zippel Lemma \cite{Schwartz1980, Zippel1979} for which no deterministic alternatives have been found in the last 40 years.
  \item Can the approximation factor be improved from $3/2 + \eps$ to $3/2$?
    We loose some approximation quality both when determining which edges to delete from the \ac{CSF}, and when solving \ac{RPP}. 
    Improving the first point would require a further refinement of the tree-rearrangement technique introduced by Xu and Rodrigues~\cite{XuR2015}.
    For the second point, the \ac{RPP} algorithm of Gutin et al. would need to be sped up to run in strongly polynomial time.
    Again, their algorithm relies on algebraic techniques for which derandomization appears difficult, so a major technical innovation for $k$-component \ac{RPP} is maybe necessary.
  \item Does there exist some polynomial-time $\alpha$-approximation algorithm for \ac{MDMTSP} with $\alpha < 2$?
    We know from Traub et al. \cite{TraubVZ2020} that any $\alpha$-approximation algorithm for single-salesperson \ac{TSP} implies a $(\alpha+\eps)$-approximation for \ac{MDMTSP} for any constant number of depots, i.e. in time $n^{\Theta(d)}$.
    For instances with many depots however, the problem remains intractable.
    It is of particular interest that two major technical tools for the classical \ac{TSP}, Christofides' Algorithm and the Subtour-Elimination LP, fail to achieve better-than-$2$-approximations in the multi-depot regime (see \cref{fig:CounterexampleChristofides} and \cref{fig:CounterexampleLP}).
    It appears that to make progress on a polynomial-time algorithm some novel structural insights would be required.
\end{itemize}
\paragraph{Acknowledgements}The third author thanks L{\'a}szl{\'o} V{\'e}gh for inspiring discussions on multi-depot TSP and feedback on an earlier version.

\bibliography{Bibliography}
\end{document}